%% file: main.tex
\documentclass[11pt, letterpaper]{article}
\usepackage[utf8]{inputenc}
\usepackage{amsfonts,amsthm,amsmath,amssymb}
\usepackage[ruled]{algorithm2e}
\usepackage{array}
\usepackage{graphicx}
\usepackage{cleveref}
\usepackage{pifont}
\usepackage{xcolor}
\usepackage{enumerate}
\usepackage{cite}
\usepackage[american]{babel}
\usepackage{geometry}
\newcommand{\anyapproxweighteddiameter}{

Given any positive function $\alpha(n)$, any algorithm (even randomized) computing an $\alpha(n)$-approximation to the weighted diameter or bi-chromatic diameter in a given graph $G$ requires $\Tilde{\Omega}(\sqrt{n}+D)$ rounds.
}

\newcommand{\anyapproxdirecteddiameter}{

Given any positive function $\alpha(n)$, any algorithm (even randomized) computing an $\alpha(n)$-approximation to the diameter in a given directed graph $G$ requires $\Tilde{\Omega}(\sqrt{n}+D)$ rounds.
}

\newcommand{\anyapproxbichromaticdirecteddiameter}{
Given any positive function $\alpha(n)$, any algorithm (even randomized) computing an $\alpha(n)$-approximation to the bi-chromatic diameter in a given directed graph $G$ requires $\Tilde{\Omega}(\sqrt{n}+D)$ rounds.
}

\newcommand{\twoapproxbichromaticdirecteddiameter}{
For all constant $\epsilon>0$, there is no $o(\frac{n}{\log ^2 n})$ round algorithm for computing a $(2-\epsilon)$-approximation to the bi-chromatic diameter in a directed graph.
}

\newcommand{\anyapproxradius}{

Given any positive function $\alpha(n)$, any algorithm (even randomized) computing an $\alpha(n)$-approximation to the weighted (directed) radius in a given weighted (directed) graph $G$ requires $\Tilde{\Omega}(\sqrt{n}+D)$ rounds.
}

\newcommand{\twoapproxradius}{
Given any constant $\varepsilon>0$, any algorithm (even randomized) computing an $(2-\varepsilon)$-approximation to the weighted (directed) radius in a given weighted (directed) graph $G$ requires $\Tilde{\Omega}(n)$ rounds.
}

\newcommand{\ssspconnection}{
For any $\eps \geq 0$, given a $(1+\eps)$-approximation algorithm $\sssp$ for weighted and directed SSSP running in $T(n,\eps,D)$ rounds,  there exists an algorithm for $(2+\eps^3 + 3\eps^2 + 4\eps)$-approximate diameter, radius, and all eccentricities in $\tilde{O}(T(n,\eps,D)+D)$ rounds on weighted, directed graphs.
}

\newcommand{\approxupperbichromaticdiameter}{
There is an algorithm with complexity $\too(\sqrt{n}+D)$ that given an undirected, unweighted graph $G=(V,E)$, and sets  $S\subseteq V,T=V\backslash S$, w.h.p. computes a value $D_{ST}^*$ such that $\frac{3D_{ST}}{5}-\frac{6}{5}\leq D_{ST}^*\leq D_{ST}$.
}

\newcommand{\approxweightedupperbichromaticdiameter}{
There is an algorithm with complexity $T(SSSP)$ that given an undirected graph $G=(V,E)$, and sets  $S\subseteq V,T=V\backslash S$, computes a value $D^*$ such that $\frac{D_{ST}}{2}-W/2\leq D^*\leq D_{ST}$. Here $W$ is the minimum edge weight in $S\times T$. 
}

\newcommand{\approxweightedlowerbichromaticdiameter}{
For all constant $\epsilon>0$, there is no $o(\frac{n}{\log ^3 n})$ round algorithm for computing a $(\frac{5}{3}-\epsilon)$-approximation to the bi-chromatic diameter in an unweighted, undirected graph.
}

\newcommand{\corolapproxone}{
For any $\eps = 1/\polylog(n)$, there exists an algorithm for $(2+\eps)$-approximate diameter, radius and all eccentricities running in $\tilde{O}(\sqrt{n}+D)$ rounds on nonnegative weighted graphs, with $n$ nodes and hop-diameter $D$.
}

\newcommand{\corolapproxtwo}{
For any $\eps = 1/\polylog(n)$, there exists an algorithm for $(2+\eps)$-approximate diameter, radius and all eccentricities running in $\too(\sqrt{n}D^{1/4}+D)$ rounds on nonnegative weighted, directed graphs, with $n$ nodes and hop-diameter $D$.
}

\newcommand{\corolapproxthree}{
There exists an algorithm for $2$-approximate radius, diameter and all eccentricities running in $\tilde{O}(\sqrt{n}D^{1/4}+D)$ rounds on nonnegative weighted, directed graphs, with $n$ nodes and hop-diameter $D$.
}

\newcommand{\approxunweightedundirected}{
For any $k\in\N$, there exist algorithms that compute $(2-\frac{1}{2^k})$-approximate diameter and radius and $(3-\frac{4}{2^k+1})$-approximate eccentricities on unweighted, undirected graphs, that have running time of $\tilde{O}(n^{\frac{1}{k+1}}+D)$ rounds w.h.p.
}

\input{preamble.tex}

\date{\today}

\begin{document}

\title{Distributed Distance Approximation}
\author{Bertie Ancona \\ \text{MIT, bancona@alum.mit.edu} \and Keren Censor-Hillel \\ \text{Technion, ckeren@cs.technion.ac.il}\and Mina Dalirrooyfard \\ \text{MIT, minad@mit.edu} \and Yuval Efron \\ \text{Technion, efronyuv@gmail.com} \and Virginia Vassilevska Williams \\ \text{MIT, virgi@mit.edu}}
\maketitle
\begin{abstract}
    \input{abstract}
\end{abstract}
\newpage
\setlength\extrarowheight{3pt}

\section{Introduction}\label{sec:-intro}
\input{Intro.tex}

\section{Preliminaries}\label{sec:-Prelim}
\input{prelim.tex}

\section{Approximation Algorithms}\label{sec:-approx}
\input{algorithms.tex}

\section{Hardness of Approximation}\label{sec:-hardnessOfApprox}
\input{lowerbounds.tex}

\section*{Acknowledgment}
This project has received funding from the European Research Council (ERC) under the European
Unions Horizon 2020 research and innovation programme under grant agreement No 755839.

\bibliographystyle{alpha}
\bibliography{main}


\end{document}

%% file: preamble.tex
\def \sssp{\mathcal{A_{\eps}}}

\newtheorem{theorem}{Theorem}[section]
\newtheorem{corollary}[theorem]{Corollary}
\newtheorem{lemma}{Lemma}[theorem]

\newtheorem{definition}[theorem]{Definition}

\newtheorem{remark}[theorem]{Remark}

\newcommand{\cut}{cut}
\newcommand{\size}[1]{\ensuremath{\left|#1\right|}}
\newcommand{\cgst}{$\mathsf{CONGEST}$\xspace}
\newcommand{\too}{\Tilde{O}}

\newenvironment{proof-sketch}{\noindent{\bf Sketch of Proof}\hspace*{1em}}{\qed\bigskip}
\newenvironment{proof-idea}{\noindent{\bf Proof Idea}\hspace*{1em}}{\qed\bigskip}
\newenvironment{proof-of-lemma}[1]{\noindent{\bf Proof of Lemma #1}\hspace*{1em}}{\qed\bigskip}
\newenvironment{proof-attempt}{\noindent{\bf Proof Attempt}\hspace*{1em}}{\qed\bigskip}

\newenvironment{definition-repeat}[1]{\begin{trivlist}
		\item[\hspace{\labelsep}{\bf\noindent Definition \ref{#1} }]\em }%
	{\end{trivlist}}

\newenvironment{lemma-repeat}[1]{\begin{trivlist}
		\item[\hspace{\labelsep}{\bf\noindent Lemma \ref{#1} }]\em }%
	{\end{trivlist}}
\newenvironment{theorem-repeat}[1]{\begin{trivlist}
		\item[\hspace{\labelsep}{\bf\noindent Theorem \ref{#1} }]\em }%
	{\end{trivlist}}

\newenvironment{corollary-repeat}[1]{\begin{trivlist}
		\item[\hspace{\labelsep}{\bf\noindent Corollary \ref{#1} }]\em }%
	{\end{trivlist}}

\makeatletter
\def\fnum@figure{{\bf Figure \thefigure}}
\def\fnum@table{{\bf Table \thetable}}
\long\def\@mycaption#1[#2]#3{\addcontentsline{\csname
  ext@#1\endcsname}{#1}{\protect\numberline{\csname
  the#1\endcsname}{\ignorespaces #2}}\par
  \begingroup
    \@parboxrestore
    \small
    \@makecaption{\csname fnum@#1\endcsname}{\ignorespaces #3}\par
  \endgroup}
\def\mycaption{\refstepcounter\@captype \@dblarg{\@mycaption\@captype}}
\makeatother

\newcommand{\mathify}[1]{\ifmmode{#1}\else\mbox{$#1$}\fi}
\newcommand{\bigO}O
\newcommand{\set}[1]{\mathify{\left\{ #1 \right\}}}

\newcommand{\remove}[1]{}
\newcommand{\ignore}[1]{}

\def\eps{\varepsilon}

\def\N{\mathbb{N}}

\newcommand{\ceil}[1]{\left\lceil #1 \right\rceil}

\def\poly{{\rm poly}}
\def\polylog{{\rm polylog}}



%% file: abstract.tex
Diameter, radius and eccentricities are fundamental graph parameters, which are extensively studied in various computational settings. Typically, computing approximate answers can be much more efficient compared with computing exact solutions. In this paper, we give a near complete characterization of the trade-offs between approximation ratios and round complexity of distributed algorithms for approximating these parameters, with a focus on the weighted and directed variants.

Furthermore, we study \emph{bi-chromatic} variants of these parameters defined on a graph whose vertices are colored either red or blue, and one focuses only on distances for pairs of vertices that are colored differently. Motivated by applications in computational geometry, bi-chromatic diameter, radius and eccentricities have been recently studied in the sequential setting [Backurs et al. STOC'18, Dalirrooyfard et al. ICALP'19]. We provide the first distributed upper and lower bounds for such problems. 

Our technical contributions include introducing the notion of \emph{approximate pseudo-center}, which extends the \emph{pseudo-centers} of [Choudhary and Gold SODA'20], and presenting an efficient distributed algorithm for computing approximate pseudo-centers. On the lower bound side, our constructions introduce the usage of new functions into the framework of reductions from 2-party communication complexity to distributed algorithms.

%% file: Intro.tex
The diameter and radius are central graph parameters, defined as the maximum and minimum eccentricities over all vertices, respectively, where the eccentricity of a vertex $v$ is the maximum distance out of $v$.
Computing the diameter and radius of a given graph are cornerstone problems with abundant applications. This is particularly the case in the context of distributed computing, where distances between nodes in a network (and in particular the graph diameter) directly influence the time it takes to communicate throughout the network. 
 
We focus on computing the diameter, radius and eccentricities in the classic \cgst model of distributed computation, in which $n$ nodes of a synchronous network communicate by exchanging messages of $O(\log{n})$ bits with their neighbors in the underlying network graph. In a seminal work, Frischknecht et al. \cite{FrischknechtHW12} showed that the diameter is hard to compute in \cgst, namely that $\Tilde{\Omega}(n)$\footnote{Throughout the paper, $\tilde{O}$ and $\tilde{\Omega}$ are used to hide poly-logarithmic factors} rounds are required, even in undirected unweighted graphs. Abboud et al.~\cite{AbboudCK16} showed that the same holds for computing the radius. Both of these results are tight up to logarithmic factors due to algorithms that compute all pairs shortest paths (APSP) in a given unweighted, undirected graph in $O(n)$ rounds, see Holzer and Wattenhofer, Lenzen and Peleg, and Peleg et al. \cite{HolzerW12,LenzenP13,PelegRT12}. Recently, Bernstein and Nanongkai \cite{Bernstein:2019:DEW:3313276.3316326}, presented an algorithm which computes exact APSP in a given weighted, directed graph in $\too(n)$ rounds as well. 

As computing the diameter and radius exactly in general graphs is hard,
 a natural relaxation is to settle for approximate computations. In an unweighted, undirected graph, a simple observation due to the triangle inequality is that computing a BFS tree from any node yields a $2$-approximation to the diameter or radius, and a $3$-approximation of all eccentricities.

Obtaining a more thorough understanding of the complexity landscape of computing approximations to these distance parameters has been an ongoing endeavour of the community. The current state of the art for diameter approximation is the algorithm by Holzer et al. \cite{HolzerPRW14} with round complexity of $O(\sqrt{n\log n}+D)$, that achieves a $\frac{3}{2}$-approximation of the diameter in a given unweighted, undirected graph (further discussion is deferred to \cref{subsection:related}).

However, many open cases have remained, and unveiling the full picture of the trade-offs between approximation ratios and round complexity for distance parameters in the \cgst model has remained a central open problem.
In this paper, we give a near-complete characterization of this trade-off for the problems of diameter, radius and eccentricities, focusing on the weighted and/or directed variants. For the problem of directed diameter, only the range $[\frac{3}{2},2]$ of approximation ratios remains open.

In some cases, originally motivated by computational geometry problems \cite{DBLP:journals/siamcomp/Yao82,DBLP:journals/ijcga/KatohI95,DBLP:journals/dcg/AgarwalES91,DBLP:journals/jgaa/DumitrescuG04}, we are interested in a ``bi-chromatic'' definition of the parameters. In the bi-chromatic setting, the vertices are partitioned into two sets, $S$ and $T=V\setminus S$, and the bi-chromatic eccentricity of a node $s\in S$ is the maximum distance from $s$ to a node in $T$. The bi-chromatic diameter and radius are the maximum and minimum bi-chromatic eccentricities of nodes in $S$.

The bi-chromatic versions of diameter and radius have received much recent attention in the sequential setting \cite{Backurs:2018:TTA:3188745.3188950,DBLP:conf/icalp/DalirrooyfardW019a}. In this paper, we initiate the study of these problems in the \cgst model, by providing upper and lower bounds for these problems. For example, we prove that a $\frac{5}{3}$-approximation to bi-chromatic diameter in an unweighted, undirected graph can be computed in $\tilde{O}(\sqrt{n}+D)$ rounds, and we prove this is tight in the sense that any improvement in the approximation ratio incurs a blowup in the round complexity to $\tilde{\Omega}(n)$.

A more comprehensive display of our results follows. Also, a comparison with previous work is depicted in Table \ref{table-upper} and Table \ref{table-lower} and is elaborated upon in Section \ref{subsection:related}.

\input{tableupper}
\input{tablelower}

\subsection{Our contributions and techniques}
\label{subsection:contribution}
As mentioned earlier, the \emph{eccentricity} $ecc(v)$ of a vertex $v$ is the distance $\max_{u\in V} d(v,u)$. The \emph{diameter} $D$ is the largest eccentricity in the graph, and the \emph{radius} $r$ is the smallest.

\subparagraph*{Directed/weighted Radius and Eccentricities.}

We present a connection between the complexity of computing or approximating the Single Source Shortest Paths (SSSP) problem and the complexity of approximating radius, diameter and eccentricities.
Formally, we prove the following theorem in Section \ref{sec:-approx}.

\begin{theorem}\label{thrmintro:-WeightedDirectedRadiusUpperBoundTwo}
\ssspconnection
\end{theorem}

We now describe the challenges in proving the above and how we cope with them.
A useful notion for distance parameters is the \emph{center} of a graph, which is the vertex with the lowest eccentricity. Given the center $c$ of a graph, we can easily approximate all eccentricities of a given graph by performing an SSSP algorithm rooted at $c$, and letting each node $v$ estimate its eccentricity by outputting $d(v,c)+ecc(c)$. However, computing the center of a graph, or even its eccentricity (the radius), is a hard task that requires $\tilde{\Omega}(n)$ rounds \cite{AbboudCK16}.

For proving Theorem \ref{thrmintro:-WeightedDirectedRadiusUpperBoundTwo}, we rely on an approach of Choudhary and Gold \cite{2eccalgo}.
Here, one defines a notion of a \emph{pseudo-center} and one then shows how to compute a pseudo-center of size $O(\log ^2 n)$ sequentially in near-linear time. A pseudo-center $C$ is a set of nodes, whose goal is to mimic the center of the graph, by promising that all eccentricities are at least the maximal distance between any node to the pseudo-center $C$. Using such a pseudo-center, one estimates the eccentricity of every node, similarly to the case of computing the actual center. 

The algorithm of \cite{2eccalgo} for computing a small pseudo-center can be viewed as a reduction to Single Source Shortest Paths (SSSP), which is very efficient in the sequential setting. However, the current state-of-the-art {\em distributed} complexity of computing exact SSSP is very costly, and hence we wish to avoid it. To overcome this, we introduce the notion of an \emph{approximate pseudo-center}, which generalizes the notion of a pseudo-center. We prove that (i) an approximate 
pseudo-center of small size can be computed efficiently in a distributed manner (thus avoiding the complexities of exact SSSP), and (ii) an approximate pseudo-center is still sufficient for approximating the required distance parameters.

From Theorem \ref{thrmintro:-WeightedDirectedRadiusUpperBoundTwo}, using the $(1+\eps)$-approximate SSSP algorithms of \cite{becker_et_al:LIPIcs:2017:8003,DBLP:conf/focs/ForsterN18}, which run in $\tilde{O}((\sqrt{n} + D)/\eps)$ rounds on weighted, undirected graphs and $\tilde{O}((\sqrt{n}D^{1/4}+D)/\eps)$ rounds on weighted, directed graphs, respectively, we deduce the following corollaries:
\begin{corollary}{}\label{corol:-eccentApprox1-intro}

\corolapproxone
\end{corollary}

\begin{corollary}\label{corol:-eccentapprox2-intro}

\corolapproxtwo
\end{corollary}

Using the exact SSSP algorithm of Chechik and Mukhtar \cite{10.1145/3382734.3405729} we obtain the following.

\begin{corollary}\label{corol:-eccentapprox3-intro}

\corolapproxthree
\end{corollary}

Regarding radius, the only previous result regarding the complexity of approximating the radius in the \cgst model is due to \cite{AbboudCK16}, in which they showed that for any $\varepsilon>0$, computing an $(3/2-\varepsilon)$-approximation to the radius in undirected, unweighted graphs requires $\Tilde{\Omega}(n)$ rounds. 
Abboud et al.~\cite{AbboudCK16} show that any algorithm computing an $(\frac{5}{3}-\varepsilon)$-approximation of all eccentricities requires $\Tilde{\Omega}(n)$ rounds as well. Having a complete understanding of the relationship between approximation ratio and the round complexity of computing unweighted, undirected radius remains an intriguing open problem.
As a step towards resolving this problem, we give a nearly full characterization of the approximation factor to round complexity mapping for radius in {\em weighted} or {\em directed} graphs in the \cgst model.

In Section \ref{sec:-hardnessOfApprox} we prove the following.
\begin{theorem}\label{thrmintro:-weightedDirectedRadiusLowerBoundTwo}
\twoapproxradius
\end{theorem}

A standard technique for proving lower bounds for the \cgst model, is to reduce it from 2-party communication complexity. In the context of the distance parameters discussed in this work, this framework was used by \cite{FrischknechtHW12} to show that any algorithm that distinguishes between networks with diameter 2 and 3 requires $\Tilde{\Omega}(n)$ rounds. Later, \cite{AbboudCK16} showed that this lower bound holds even when one considers sparse networks with only $O(n)$ edges (they also proved more results as discussed in the related work section). 
 
Many of the papers that employ this framework, reduce from either the Set Disjointness function, the Equality function, or the Gap Disjointness function \cite{DasSarma:2011:DVH:1993636.1993686,Censor-HillelD18,CzumajK18,BachrachCDELP19}. In this work, we enhance this framework by showing lower bounds using reductions from other functions, which were not used previously to obtain lower bounds for the \cgst model. Namely, in the proof of Theorem \ref{thrmintro:-weightedDirectedRadiusLowerBoundTwo}, we use the Tribes function, defined by Jayram et al. in \cite{DBLP:conf/stoc/JayramKS03}, and the Hitting Set Existence (HSE) function, which is a communication complexity variant of a problem introduced by Abboud et al. in \cite{Abboud:2016:AFP:2884435.2884463}. We elaborate upon this framework and the functions that we use in Section \ref{sec:-Prelim}.

The following is a corollary of Theorem \ref{thrmintro:-weightedDiameterLowerBoundAny} and Theorem \ref{thrmintro:-unweightedDirectedDiameterLowerBoundAny} which are stated below for the diameter, since any finite approximation to the radius, implies a finite approximation to the diameter, as $r\leq D\leq 2r$.

\begin{corollary}\label{thrmintro:-weightedDirectedRadiusLowerBoundAny}
\anyapproxradius
\end{corollary}

\subparagraph*{Directed/Weighted Diameter.} In previous work, Holzer and Pinsker \cite{holzer_et_al:LIPIcs:2016:6597} showed a lower bound of $\Tilde{\Omega}(n)$ rounds for computing a $(2-\varepsilon)$-approximation of the diameter of a given weighted graph.
Shortly after, Becker et al. \cite{becker_et_al:LIPIcs:2017:8003} designed an algorithm that computes a $(2+o(1))$-approximation of weighted and directed diameter in $\Tilde{O}(\sqrt{n}D^{1/4}+D)$ rounds.
Such an algorithm makes one wonder, is there a smooth trade-off between the round complexity and the approximation ratio when going beyond a $2$-approximation, for either the directed or weighted variants? In other words, can one further reduce the round complexity if we are willing to settle for a worse approximation ratio? For  weighted diameter, this question was resolved by Lenzen et al. \cite{DBLP:journals/dc/LenzenPP19} in the negative, in the sense that the dependence on $n$ in the algorithm of \cite{becker_et_al:LIPIcs:2017:8003} is necessary (up to poly-logarithmic factors) for any approximation of the diameter in weighted or directed graphs. We give a proof of this result for completeness, and this allows us to more easily present a similar new result for the \emph{bi-chromatic} diameter case. The bi-chromatic diameter is a variant of the diameter problem that is discussed later.

\begin{theorem}\label{thrmintro:-weightedDiameterLowerBoundAny}
\anyapproxweighteddiameter
\end{theorem}

\begin{theorem}\label{thrmintro:-unweightedDirectedDiameterLowerBoundAny}
\anyapproxdirecteddiameter
\end{theorem}

To prove these theorems we reduce from the problem of Spanning Connected Subgraph Verification (SCSV) to approximating these parameters. The SCSV problem is known to admit the above lower bound due to Das Sarma et al.~\cite{DasSarma:2011:DVH:1993636.1993686}. The key challenge is to construct a reduction in a manner that can be efficiently simulated in \cgst. 
The proofs are given in Section \ref{ssec:any_approx_lower_bound}.

\subparagraph*{Undirected and Unweighted Diameter, Radius and Eccentricities.} 
Abbout et al. \cite{AbboudCK16} show that for any $\varepsilon>0$, any algorithm computing an $(\frac{3}{2}-\varepsilon)$-approximation of diameter or radius in unweighted undirected graphs has round complexity $\Tilde{\Omega}(n)$. Furthermore, any algorithm computing an $(\frac{5}{3}-\varepsilon)$-approximation to all eccentricities has round complexity $\Tilde{\Omega}(n)$. For upper bounds, the state of art for diameter approximation is an algorithm by Holzer et al. \cite{HolzerPRW14}, computing a $3/2$-approximation in $\Tilde{O}(\sqrt{n\log n}+D)$ rounds.
Fully understanding the mapping of approximation ratios in the range $[\frac{3}{2},2)$ for diameter and radius, and in the range $(\frac{5}{3},3)$ for all eccentricities, to their respective correct round complexity in the \cgst model remains open. As a step towards resolving this open problem, we present a simple distributed implementation of a sequential approximation algorithm of Cairo et al. \cite{Cairo:2016:NBA:2884435.2884462} for diameter, radius and eccentricities with the following parameters.

\begin{theorem}\label{thm:approx_unweighted_undirected}
\approxunweightedundirected
\end{theorem}

\subparagraph*{Bi-chromatic Diameter and Radius.} To the best of our knowledge, no previous results regarding bi-chromatic distance parameters are known in distributed settings. Roughly speaking, these variants are defined using only distances between pairs of nodes in $S\times T$ where $S,T\subseteq V, T=V\setminus S$. $D_{ST}$,$R_{ST}$ respectively denote the $ST$-diameter $\max_{s\in S,t\in T} d(s,t)$ and the $ST$-radius $\min_{s\in S}\max_{t\in T} d(s,t)$ (also see Section \ref{subsec:prelim}).
In the following, proven in Section \ref{subsec:ApxST}, $T(SSSP)$ refers to the distributed complexity of exact weighted SSSP.

\begin{theorem}\label{thm:approx_bichro_diameter_unweighted_intro}

\approxupperbichromaticdiameter
\end{theorem}

\begin{theorem}\label{thm:approx_bichro_diameter_weighted_intro}
\approxweightedupperbichromaticdiameter
\end{theorem}

We remark that using very similar algorithms to the ones of \cref{thm:approx_bichro_diameter_unweighted_intro} and \cref{thm:approx_bichro_diameter_weighted_intro}, one can obtain the following results, whose proofs we omit due to similarity to the main ideas in the proofs we provide for the above two theorems. 
\begin{remark}\label{remark:additional_results_unweighted}
There are algorithms with complexity $\too(\sqrt{n}+D)$ that given an undirected, unweighted graph $G=(V,E)$, and sets  $S,T\subseteq V$, compute w.h.p. the following.
\begin{enumerate}
    \item  A value $R_{ST}^*$ such that $R_{ST}\leq R_{ST}^*\leq\frac{5R_{ST}}{3}+\frac{5}{3}$, in the case that $S=V\backslash T$.
    \item A $2$-approximation to all $ST$-eccentricities.
    \item A $2$-approximation to $R_{ST}$.
\end{enumerate}
\end{remark}

\begin{remark}\label{remark:additional_results_weighted}
There are algorithms with complexity $T(SSSP)$ that given an undirected graph $G=(V,E)$, and sets $S,T\subseteq V$, compute the following.
\begin{enumerate}
    \item A value $R_{ST}^*$ such that $R_{ST}\leq R_{ST}^*\leq 2R_{ST}+W$, in the case that $S=V\backslash T$. Here $W$ is the minimum edge weight in $S\times T$.
    \item A $3$-approximation to all $ST$-eccentricities.
    \item A $3$-approximation to $R_{ST}$.
\end{enumerate}
\end{remark}

We complement these upper bounds with several lower bounds. We show that in the weighted case, one cannot hope to do better than a $\frac{5}{3}$-approximation for bi-chromatic diameter with $O(n^{1-\epsilon})$ rounds for some $\epsilon>0$.
Additionally, as a step towards realizing the complexity of finding a better than 2-approximation for directed diameter, we show that for \emph{bi-chromatic} diameter, in which one is tasked with finding the largest distance between a pair of nodes in different sets of a given partition of the graph, finding such an approximation is a hard task. 

Formally, we prove the following theorems in Section \ref{ssec:-bichromaticD_hard}.

\begin{theorem}\label{thm:lowerbound_diameter_bichromatic_intro}

\approxweightedlowerbichromaticdiameter
\end{theorem}

\begin{theorem}\label{thm:lowerbound_diameter_bichromatic_directed_intro}
\twoapproxbichromaticdirecteddiameter
\end{theorem}

Finally, we show that for both the directed and weighted cases, any approximation of the bi-chromatic diameter requires $\tilde{\Omega}(\sqrt{n}+D)$ rounds. The weighted case is proved as part of \cref{thrmintro:-weightedDiameterLowerBoundAny}. In Section \ref{ssec:any_approx_lower_bound}, we prove separately the directed case, which is stated formally as follows.

\begin{theorem}\label{thm:bichromatic_directed_diameter_lowerbound_any}
\anyapproxbichromaticdirecteddiameter
\end{theorem}

\subsection{Additional related work}
\label{subsection:related}

The state of the art algorithm for $3/2$-approximation of unweighted, undirected diameter \cite{HolzerPRW14} was preceded by a significant number of works. Notable examples are Holzer's and Wattenhofer's algorithm computing a $3/2$-approximation of the diameter in undirected, unweighted graphs in $O(n^{3/4}+D)$ rounds \cite{HolzerW12}, and the independent work of Peleg et al. \cite{PelegRT12}, which achieves the same approximation in $O(D\sqrt{n}\log n)$ rounds.
Later, Lenzen and Peleg \cite{LenzenP13} improved this upper bound to $O(\sqrt{n}\log n+D)$.

Approximations to more concrete variants of distance computations such as APSP and SSSP have been extensively studied in the \cgst as well. Examples include the deterministic  $(1+o(1))$-approximation to APSP by Nanongkai \cite{DBLP:conf/stoc/Nanongkai14}, and the $(1+\epsilon)$-approximation algorithm for SSSP of Becker et al. \cite{becker_et_al:LIPIcs:2017:8003}. The near optimal algorithm of Bernstein and Nanongkai for APSP \cite{Bernstein:2019:DEW:3313276.3316326} was preceded by a series of papers that set to realize the complexity of APSP in \cgst \cite{DBLP:conf/focs/HuangNS17,DBLP:conf/stoc/Elkin17,DBLP:conf/ipps/AgarwalR19,DBLP:conf/podc/AgarwalRKP18,DBLP:journals/corr/abs-1810-08544}. Given that \cite{Bernstein:2019:DEW:3313276.3316326} is a randomized Las Vegas algorithm, there remains a gap between between the best known deterministic and randomized algorithms for APSP, with the deterministic state of the art being $\tilde{O}(n^{4/3})$ \cite{DBLP:conf/spaa/AgarwalR20}. For SSSP, the state of the art algorithm of \cite{10.1145/3382734.3405729} was also preceded by a series of improvements \cite{DBLP:conf/stoc/Elkin17,DBLP:conf/stoc/Nanongkai14,becker_et_al:LIPIcs:2017:8003,DBLP:conf/stoc/LenzenP13, DBLP:conf/stoc/GhaffariL18,DBLP:conf/stoc/HenzingerKN16,DBLP:journals/siamcomp/ElkinN19,DBLP:conf/focs/ForsterN18} from the folklore $O(n)$ Bellman-Ford algorithm.

Approximations to distance computations have been studied in various distributed settings, such as the congested clique model. Starting from \cite{DBLP:journals/dc/Censor-HillelKK19}, which presented the first non trivial algorithms for both exact, and approximated APSP in the model. From there a series of works designed more and more efficient algorithms for approximating distances in the model \cite{DBLP:conf/wdag/ChechikM19,becker_et_al:LIPIcs:2017:8003,DBLP:conf/podc/Censor-HillelDK19,DBLP:conf/opodis/DinitzN19,DBLP:journals/siamcomp/ElkinN19,DBLP:conf/spaa/ElkinN19,DBLP:conf/wdag/Gall16}, with the most recent work being the $poly(\log \log n)$ approximations for APSP and Multi Source Shourtest Paths \cite{DBLP:journals/corr/abs-2003-03058}.

Conditional hardness results for these parameters are very well-studied in the sequential setting, within fine-grained complexity, under assumptions such as the Strong Exponential Time Hypothesis (SETH) \cite{DBLP:journals/jcss/ImpagliazzoP01}. For details, see e.g., the work of Backurs et al. \cite{Backurs:2018:TTA:3188745.3188950} or the survey by Vassilevska Williams \cite{vsurvey}. Returning to the \cgst model,
in some topologies such as planar graphs, work by Li and Parter \cite{LiP19} showed that the diameter of an unweighted, undirected graph can even be computed in a sublinear number of rounds.

The lower bound framework for reducing 2-party communication complexity to \cgst was introduced by Peleg and Rubinovich in \cite{PelegR00}, in which they show that any algorithm solving the minimum spanning tree (MST) problem has round complexity $\Tilde{\Omega}(\sqrt{n}+D)$. Since then, there has been a surge of lower bounds for the \cgst model employing this framework; examples include \cite{DasSarma:2011:DVH:1993636.1993686,DruckerKO13,FischerGKO18,DBLP:journals/corr/abs-1901-01630,CzumajK18,BachrachCDELP19}.
In an independent concurrent work, \cite{GKP20} show another angle of the landscape of the complexity of diameter approximation, proving that for any constant $\epsilon>0$, any algorithm approximating the diameter of a given unweighted, undirected graph, within a factor of $(\frac{3}{5}+\epsilon)$, $(\frac{4}{7}+\epsilon)$, or $(\frac{6}{11}+\epsilon)$, must have a round complexity of at least $\tilde{\Omega}(n^{1/3})$, $\tilde{\Omega}(n^{1/4})$, or $\tilde{\Omega}(n^{1/6})$, respectively.

%% file: tableupper.tex
\renewcommand{\arraystretch}{.9}
\begin{table}[h!]
\begin{tabular}{|c|c|c|c|c|}
\hline
\textbf{Problem} & \textbf{Approx.} & \textbf{Variant} & \textbf{Upper Bound $\tilde{O}(\cdot)$} & \textbf{Reference} \\
\hline {Diameter} & Exact & wted dir & $n$ & \cite{Bernstein:2019:DEW:3313276.3316326} \\
\cline{2-5} & $2-\frac{1}{2^k}$ & & $n^{\frac{1}{k+1}}+D$ & Theorem \ref{thm:approx_unweighted_undirected}, \cite{32Diam}${}^*$ \\
\cline{2-5} & $2$ & wted dir & $T(SSSP)$ & Corollary \ref{corol:-eccentapprox3-intro} \\
\cline{2-5} & $2+\epsilon$ & wted & $\sqrt{n}+D$ & \cite{becker_et_al:LIPIcs:2017:8003} \\
\cline{3-5} & & wted dir & $\sqrt{n}D^{1/4}+D$ & Corollary \ref{corol:-eccentapprox2-intro} \\
\hline {Radius} & Exact & wted dir & $n$ & \cite{Bernstein:2019:DEW:3313276.3316326} \\
\cline{2-5} & $2-\frac{1}{2^k}$ & & $n^{\frac{1}{k+1}}+D$ & Theorem \ref{thm:approx_unweighted_undirected} \\
\cline{2-5} & $2$ & wted dir & $T(SSSP)$ & Corollary \ref{corol:-eccentapprox3-intro} \\
\cline{2-5} & $2+\epsilon$ & wted & $\sqrt{n}+D$ & Corollary \ref{corol:-eccentApprox1-intro} \\
\cline{3-4}\cline{5-5} & & wted dir & $\sqrt{n}D^{1/4}+D$ & Corollary \ref{corol:-eccentapprox2-intro} \\
\hline {Eccentricities} & Exact & wted dir & $n$ & \cite{Bernstein:2019:DEW:3313276.3316326} \\
\cline{2-5} & $3-\frac{4}{2^k+1}$ & & $n^{\frac{1}{k+1}}+D$ & Theorem \ref{thm:approx_unweighted_undirected} \\
\cline{2-5} & $2$ & wted dir & $T(SSSP)$ & Corollary \ref{corol:-eccentapprox3-intro} \\
\cline{2-5} & $2+\epsilon$ & wted & $\sqrt{n}+D$ & Corollary \ref{corol:-eccentApprox1-intro} \\
\cline{3-4}\cline{5-5} & & wted dir & $\sqrt{n}D^{1/4}+D$ & Corollary \ref{corol:-eccentapprox2-intro}\\
\hline {Bi-chromatic Diameter} & Exact & wted dir & $n$ & \cite{Bernstein:2019:DEW:3313276.3316326} \\
\cline{2-5}  & $5/3$ & & $\sqrt{n}+D$ & Theorem \ref{thm:approx_bichro_diameter_unweighted_intro} \\
\cline{2-5} & $2$ & wted & $T(SSSP)$ & Theorem \ref{thm:approx_bichro_diameter_weighted_intro} \\
\hline
\end{tabular}
\vspace{2mm}
\caption{Upper bounds for the problems considered in this paper. A variant can be weighted, directed, both, or neither. Upper bounds hold for the listed variants and all subsets of those variants. Approximation factors are multiplicative but may omit additive error. The value $k$ may be any integer greater or equal to 1. We denote the round complexity of the current best exact weighted SSSP algorithm by $T(SSSP)$, currently $\tilde{O}(\min\set{\sqrt{nD},\sqrt{n}D^\frac{1}{4}+n^\frac{3}{5}} + D)$ by \cite{DBLP:conf/focs/ForsterN18}.  
${}^* \text{for } k=1.$
}
\label{table-upper}
\end{table}
\renewcommand{\arraystretch}{1}

%% file: tablelower.tex
\renewcommand{\arraystretch}{.9}
\begin{table}[h!]
\begin{tabular}{|c|c|c|c|c|}
\hline
\textbf{Problem} & \textbf{Approx.} & \textbf{Variant} & \textbf{Lower Bound $\tilde{\Omega}(\cdot)$} & \textbf{Reference} \\
\hline {Diameter} & $3/2-\eps$ & & $n$  & \cite{AbboudCK16} \\
\cline{2-3}\cline{5-5} & $2-\eps$ &  wted & & \cite{holzer_et_al:LIPIcs:2016:6597} \\ 
\cline{2-5} & $\poly(n)$ &  wted & $\sqrt{n}+D$ & \cite{DBLP:journals/dc/LenzenPP19} \\
\cline{3-3}\cline{5-5} & & dir & & Theorem \ref{thrmintro:-unweightedDirectedDiameterLowerBoundAny} \\
\hline {Radius} & $3/2-\eps$ & & $n$  & \cite{AbboudCK16} \\
\cline{2-3}\cline{5-5} & $2-\eps$ &  wted & & Theorem \ref{thrmintro:-weightedDirectedRadiusLowerBoundTwo} \\
\cline{3-3} & & dir & &  \\
\cline{2-5} & $\poly(n)$ &  wted & $\sqrt{n}+D$ & Corollary \ref{thrmintro:-weightedDirectedRadiusLowerBoundAny} \\
\cline{3-3} & & dir & & \\
\hline {Eccentricities} & $5/3-\eps$ & & $n$ & \cite{AbboudCK16} \\
\cline{2-3}\cline{5-5} & $2-\eps$ &  wted & & \cite{holzer_et_al:LIPIcs:2016:6597} \\
\cline{3-3}\cline{5-5} & & dir & & Theorem \ref{thrmintro:-weightedDirectedRadiusLowerBoundTwo} \\
\cline{2-5} & $\poly(n)$ &  wted & $\sqrt{n}+D$ & Corollary \ref{thrmintro:-weightedDirectedRadiusLowerBoundAny} \\
\cline{3-3} & & dir & & \\
\hline {Bi-chromatic Diameter} & $5/3-\eps$ &  & $n$ & Theorem \ref{thm:lowerbound_diameter_bichromatic_intro} \\
\cline{2-3}\cline{5-5} & $2-\eps$ & wted & & \cite{holzer_et_al:LIPIcs:2016:6597} \\
\cline{3-3}\cline{5-5} & & dir & & Theorem \ref{thm:lowerbound_diameter_bichromatic_directed_intro} \\
\cline{2-5} & $\poly(n)$ &  wted & $\sqrt{n}+D$ & Corollary \ref{thrmintro:-weightedDirectedRadiusLowerBoundAny} \\
\cline{3-3} & & dir & & \\
\hline
\end{tabular}
\vspace{2mm}
\caption{Lower bounds for the problems considered in this paper. A variant can be weighted, directed, both, or neither. Lower bounds hold for the listed variants and all supersets of those variants. Approximation factors are multiplicative.}
\label{table-lower}
\end{table}
\renewcommand{\arraystretch}{1}

%% file: prelim.tex
\subsection{The Model \& Definitions}\label{ssec:-model}
\label{subsec:prelim}
This paper considers the \cgst model of computation. In this model, a synchronized network of $n$ nodes is represented by an undirected, unweighted, simple graph $G=(V,E)$. In each round, each node can send a different message of $O(\log n)$ bits to each of its neighbors.

Next, we define the network parameters that we discuss in the paper.
\begin{definition}\label{def:-prameters}
Given a weighted, directed graph $G=(V,E)$, denote by $d(u,v)$ the weight of the lightest directed path starting at node $u$ and ending at node $v$. If there is no such path, we define $d(u,v)=\infty$. Here, the weight of a path $P$ is the sum of the weights of its edges. The eccentricity $ecc(u)$ of a node $u$ is defined to be $\max\limits_{v\in V} d(u,v)$. The radius $r$ of $G$ is defined to be $\min\limits_{v\in V} ecc(v)$. The diameter $D$ of $G$ is defined to be $\max\limits_{v\in V} ecc(v)$.

\end{definition}

The $ST$ variants of these distance parameters are defined as follows.

\begin{definition}[$ST$ and bi-chromatic diameter, radius and eccentricities.]\label{def:ST_diameter}
 Given a weighted graph $G=(V,E)$, and two non empty subsets $S,T\subseteq V$,
 given $v\in S$, we define its $ST$-eccentricity by $ecc(v)=\max\limits_{u\in T} d(v,u)$.
 We define the $ST$-diameter of $G$ to be $D_{ST}=\max\limits_{v\in S} ecc(v)$. The $ST$-radius of $G$ is defined to be $R_{ST}=\min\limits_{v\in S} ecc(v)$. 
 When $S=V\backslash T$, the $ST$ parameters are called bi-chromatic.
\end{definition}

\subsection{The Communication Complexity  Framework} The high level idea of applying the framework of reductions from 2-party communication complexity to obtain lower bounds in the \cgst model is as follows. We pick some function $f:\set{0,1}^k\times \set{0,1}^k\to \set{0,1}$, and then reduce any efficient communication protocol for it to an efficient \cgst algorithm for the discussed problem. We start with our two players Alice (A) and Bob (B), each of them respectively receives a binary string of length $k$ denoted by $x,y\in \set{0,1}^k$.

We construct a graph $G=(V,E)$ we call the \emph{fixed graph construction}, and we partition the set of vertices $V$ into the sets $V_A,V_B$. We call the cut induced by $V_A,V_B$ the \emph{communication cut}, and we denote the number of edges in this cut by $|cut|$.

Now, given $x$ and the graph $G[V_A]$ (i.e., the subgraph of $G$ induced by $V_A$), Alice modifies the graph $G[V_A]$ in any way that may depend only on $x$, and Bob does the same with $y$ and $G[V_B]$. Denote the resulting graph by $G_{x,y}$, and denote its number of nodes by $n$.

The resulting graph $G_{x,y}$ should be constructed such that it has some property $P$ (e.g. radius at least 3) iff $f(x,y)=1$. Now, assuming there is an algorithm $Alg$ in the \cgst model that decides $P$ in $T$ rounds, Alice and Bob can simulate this algorithm on $G_{x,y}$, and the only communication required between them is for simulating messages that are sent on edges in the communication cut. 
Thus, Alice and bob can simulate $Alg(G_{x,y})$ while communicating $O(T\cdot |cut|\cdot \log n)$ bits of communication. Furthermore, by the property of $G_{x,y}$, deciding $P$ on $G_{x,y}$ allows them to compute $f(x,y)$ with  $O(T\cdot |cut|\cdot \log n)$ bits of communication.
Therefore, a lower bound on the communication complexity of $f$, implies a lower bound on $T$, which is the round complexity of the distributed algorithm.

We next elaborate on the functions $f$ that we use in our reductions.

\begin{definition}[The Set Disjointness Problem (Disj) \cite{RAZBOROV1992385}]
\label{def:DISJ_Def}
Alice and Bob receive subsets $X,Y\subseteq[n]$, respectively, represented as binary vectors of length $n$. Their goal is to decide whether $X\cap Y=\emptyset$.
\end{definition}

It is known by \cite{DBLP:journals/jcss/Bar-YossefJKS04,Kushilevitz:1996:CC:264772,RAZBOROV1992385} that the randomized communication complexity of Disj on inputs of size $n$ is $\Omega(n)$.

\begin{definition}
[The Tribes (ListDISJ) Problem \cite{DBLP:conf/stoc/JayramKS03}] 
\label{def:-LSDISJ}
Alice and Bob are given sets $A_i,B_i\in\set{0,1}^N$ for each $i\in [N]$. They must output 1 if and only if there is some $i$ such that $A_i$ and $B_i$ are disjoint, i.e. there is no $j$ such that $A_{ij}=B_{ij}=1$. We treat the inputs $x$ and $y$ as binary strings of length $N^2$, such that $x=A_1\circ ...\circ A_N,y=B_1\circ...\circ B_N$. Here, $\circ$ refers to string concatenation.
\end{definition}

The Tribes function is defined in \cite{DBLP:conf/stoc/JayramKS03}, where a lower bound of $\Omega(N^2)$ communication bits is proved, even for randomized protocols. 


\begin{definition}[Orthogonal vectors (OV)]
\label{def:OV_communication}
Alice and Bob each receive $N$ binary vectors of length $d$, namely, $X,Y\subseteq \set{0,1}^d,|X|=|Y|=N$ for some $d$. It holds that $OV(X,Y)=1$ iff there exists $v\in A,u\in B$ such that $u\cdot v=0$, where $\cdot$ denotes the standard inner product over the reals.
\end{definition}

It is known by \cite{BRINGMANN201810} that the randomized communication complexity of OV on inputs of size $N$ and $d=\Theta(\log N)$ is $\Omega(N)$.


\begin{definition}[The Hitting Set Existence (HSE) Problem \cite{Abboud:2016:AFP:2884435.2884463}]\label{def:-HSE}
Alice and Bob are given sets $A$ and $B$ of $N$ Boolean vectors of size $d$. They must output 1 if and only if there is some vector $a\in A$ such that for all vectors $b\in B$ it holds that $a\cdot b \neq 0$.
\end{definition}

The HSE problem was used in \cite{Abboud:2016:AFP:2884435.2884463} for showing fine-grained complexity lower bounds, and conjectured to be sequentially hard. Here, we prove by reducing the Disj problem to the HSE problem that HSE is hard in the 2-party communication setting.
\begin{theorem}
\label{thrm:-HSEhard}
The HSE problem on sets of size $N$ of vectors of size $d=2\log{N}+1$ requires $\Omega(N)$ bits of communication.
\end{theorem}
\begin{proof}
We reduce from the Set Disjointness problem.
Let $X,Y\in\set{0,1}^N$ be the sets received by Alice and Bob respectively. Each may convert these sets into vector sets $A$ and $B$ as follows.

Let $b(i)$ be the bit representation of the number $i$, and let $\overline{b}(i)$ be the bitwise complement of $b(i)$. For each $i$ in $[n]$, let $A_i$ be the concatenation of $b(i)$, $\overline{b}(i)$, and $X_i$, where $X_i$ is the $i$th element of $X$. For each $i$ in $[N]$, let $B_i$ be the concatenation of $\overline{b}(i)$, $b(i)$, and $Y_i$, where $Y_i$ is the $i$th element of $Y$. Let $A = \set{A_i}_i$ and $B = \set{B_i}_i$.

We claim that $X$ and $Y$ are disjoint if and only if there does not exist a hitting set in the instance $(A,B)$. If $X$ and $Y$ are disjoint, then for all $i$, $A_i\cdot B_i = 0$, because $b(i)$ and $\overline{b}(i)$ are orthogonal by definition and either $X_i$ or $Y_i$ is 0. Thus there is no hitting set.

Otherwise, there is some $i$ such that $X_i = Y_i = 1$. Note that for all $i\neq j$, $A_i\cdot B_j \neq 0$. Also, $A_i$ and $B_i$ are both 1 in the last bit, by our construction. Thus for all $j$, $A_i\cdot B_j \neq 0$, and $A_i$ is a hitting set.

No additional communication is required for the reduction, so $\Omega(N)$ bits are still required due to the lower bound for the Set Disjointness problem.
\end{proof}


%% file: algorithms.tex
\subsection{Approximations for weighted directed variants}
\label{subsec:ApxWD}
In this section, we prove our approximation algorithms, starting with the connection between the complexity of $SSSP$ and approximating distance parameters. Formally, we prove the following theorem, and then we deduce Corollaries \ref{corol:-eccentApprox1-intro}, \ref{corol:-eccentapprox2-intro}, and \ref{corol:-eccentapprox3-intro}.

\begin{theorem-repeat}{thrmintro:-WeightedDirectedRadiusUpperBoundTwo}
\ssspconnection
\end{theorem-repeat}

We briefly remind the reader of the discussion in the introduction regarding the theorem.
In order to obtain fast algorithms and maintaining the quality of the approximation, we generalize the notion of \emph{pseudo-center} defined by Choudhary and Gold \cite{2eccalgo} into \emph{approximate pseudo-center}. We show how to compute such a set of small size, and we show that such a set suffices to obtain the approximations detailed in Theorem \ref{thrmintro:-WeightedDirectedRadiusUpperBoundTwo}.

\begin{definition}
A $\alpha$-approximate pseudo-center is a set $C$ of nodes such that for all nodes $v\in V$, $ecc(v) \geq \max_{u\in V}\min_{c\in C}\set{d(c, u)/\alpha}$.
\end{definition}

We begin by showing that we can compute a small approximate pseudo-center efficiently.

\begin{lemma}
\label{lem:sssptopscenter}
Given a $(1+\eps)$-approximate, $T(n, \eps, D)$-round SSSP algorithm $\sssp$, there is a Las Vegas algorithm to compute a $(1+\eps)^2$-approximate pseudo-center of size $O(\log^2(n))$ of a graph $G=(V,E)$ in $\tilde{O}(T(n,\eps,D))$ rounds of communication, with high probability.
\end{lemma}
\begin{proof}
Let the set $C$ begin empty, and let $W$ begin as the set $V$. Throughout the proof, running $\sssp$ outward (inward) from a vertex $v\in V$ means computing the distances \emph{from} $v$ to the rest of the nodes (\emph{to} $v$ from the rest of the nodes). We repeat the following until $W$ is empty: \begin{itemize}
    \item Assign each node in $W$ to a set $S$ independently with probability $\min\set{1, 24\log(n)/|W|}$. Resample if $|S|<8\log{n}$ or $|S|>36\log{n}$.
    \item Run $\sssp$ outward from each node in $S$, and for all $u\in V$, compute estimated distances $d_{\sssp}(S, u) = \min_{s\in S}\set{d_{\sssp}(s, u)}$.
    \item Let $a$ be the node with the largest estimated distance from $S$. Then, we broadcast $d_{\sssp}(S, a)$ to all nodes in the graph using some BFS tree.
    \item Run $\sssp$ inward from $a$, and remove all nodes $u$ where $d_{\sssp}(u, a) \geq d_{\sssp}(S, a)$ from $W$.
    \item Add $S$ to $C$.
\end{itemize}

First, we argue that $C$ is a $(1+\eps)^2$-approximate pseudo-center. We only remove a node $u$ from $W$ when $d_{\sssp}(u, a) \geq d_{\sssp}(S, a)$ for some sample $S$. Let $a^*$ be the node that is truly farthest from $S$; then $d_{\sssp}(S, a) \geq d(S,a^*)/(1+\eps) \geq \max_{x\in V}\min_{c\in C}\set{d(c, x)/(1+\eps)}$, because $S\subseteq C$. We also note that by similarly bounding the error of $\sssp$,  it holds that $d_{\sssp}(u, a) \leq (1+\eps)d(u,a) \leq (1+\eps)ecc(u)$, so we may conclude that $$(1+\eps)ecc(u) \geq \max_{x\in V}\min_{c\in C}\set{d(c, x)/(1+\eps)}.$$ In other words, $ecc(u) \geq \max_{x\in V}\min_{c\in C}\set{d(c, x)/(1+\eps)^2}$, which meets the definition of a $(1+\eps)^2$-pseudo-center.

Next, we argue that each iteration requires $\tilde{O}(T(n,\eps,D))$ rounds. Using a Chernoff bound, it is simple to show that in each round, $8\log{n}\leq |S|\leq 36\log{n}$ with probability at least $1-1/n^4$, so we expect to resample a sub-constant number of times. We then run $\sssp$ from each node in $S$ and we run it again once to the node $a$, for a total of $O(\log{n}\cdot T(n,\eps,D))$ rounds. The rest of each iteration involves a constant number of broadcasts that take $O(D)$ rounds in total.

Finally, we argue that with high probability, we only have $O(\log{n})$ iterations in our algorithm. We do this by showing that in iteration $i$, the size of $W$ reduces by at least half with high probability, i.e. $|W_i|/2\geq |W_{i+1}|$. Consider the set $X\subseteq W_i$ of $|W_i|/2$ nodes with the smallest $d_{\sssp}(u,a)$, $u\in W_i$. Note that $S_i$ is a randomly sampled subset of $W_i$ of size at least $8\log{n}$, and thus intersects $X$ with probability at least $(1-1/n^5)$, as argued in Lemma \ref{lemma:helper_lemma} below \cite{2eccalgo} with no further assumptions. 

All nodes in $W_i\backslash X$ are at least as far as any node in that intersection under $\sssp$, by definition. This implies that for all $u\in W_i\backslash X$, $d_{\sssp}(u, a) \geq d_{\sssp}(S, a)$, which implies that all $|W_i|/2$ nodes of $W_i\backslash X$ will be removed from $W_i$ in iteration $i$.
\end{proof}

\begin{lemma}[Lemma 2.1 in \cite{2eccalgo}]\label{lemma:helper_lemma}
Let $U$ be a universe set of size at most $n$, and let $S_1,...,S_n\subseteq U$ such that $|S_i|\geq L$ for each $i\in [n]$. Let $c$ be some constant and $r=\frac{n(c+1)\ln n}{L}$. Let $S\subseteq U$ be a random subset of size $r$, then it holds that $S\cap S_i\neq \emptyset$ for all $i$ with probability $1-n^{-c}$.
\end{lemma}

Now that we showed how to compute an approximate pseudo-center, we show that it is sufficient for approximating the distance parameters as claimed.

\begin{lemma}
\label{lem:pscentertoecc}
Given a $(1+\eps)^2$-approximate pseudo-center $C$ and a $(1+\eps)$-approximate SSSP algorithm $\sssp$ taking $T(n,\eps,D)$ rounds, we may compute $(2+\eps^3 + 3\eps^2 + 4\eps)$-approximate eccentricities for all nodes in $O(|C|\cdot T(n,\eps,D)+D)$ rounds.
\end{lemma}
\begin{proof}
First, we run $\sssp$ to and from each node in $C$, so that each node $v\in V$ stores $d_{\sssp}(c,v)$ and $d_{\sssp}(v,c)$ for all $c\in C$. Each node $u$ internally determines 
$\min_{c\in C}\set{d_{\sssp}(c, u)}$. Then, using aggregation over a BFS tree, the nodes determine, and then broadcast the value $D_{\sssp}(C) := \max_{u\in V}\min_{c\in C}\set{d_{\sssp}(c, u)}$. Thus, the aggregation takes $O(D)$ rounds.
Each node $v$ approximates its eccentricity as $\max_{c\in C}\set{d_{\sssp}(v, c)} + D_{\sssp}(C)$.

First, note that this estimate is at least the true eccentricity of $v$, as each computed distance represents some path in the graph, and in this distance a path can go from $v$ to any node in $C$ and then any node in $V$.

We argue that this is a $(2+\eps^3 + 3\eps^2 + 4\eps)$-approximation. The estimated distance $\max_{c\in C}\set{d_{\sssp}(v, c)}$ is at most $(1+\eps)\cdot ecc(v)$, because $\sssp$ overestimates by at most a factor of $1+\eps$. By our definition of $(1+\eps)^2$-approximate pseudo-center, $D(C)\leq (1+\eps)^2ecc(v)$. Our estimate $D_{\sssp}(C)$ is at most $(1+\eps)\cdot D(C)$, so $D_{\sssp}(C)\leq (1+\eps)^3ecc(v)$. Thus, $\max_{c\in C}\set{d_{\sssp}(v, c)} + D_{\sssp}(C) \leq (1+\eps+(1+\eps)^3)\cdot ecc(v) = (2+\eps^3 + 3\eps^2 + 4\eps)\cdot ecc(v)$.

We compute $\sssp$ twice for each element of $C$, and broadcast a constant number of values to all nodes, so the total number of rounds is $O(|C|\cdot T(n,\eps,D)+D)$.
\end{proof}

\begin{proof}[Proof of Theorem \ref{thrmintro:-WeightedDirectedRadiusUpperBoundTwo}]$ $
Applying Lemma \ref{lem:sssptopscenter} and Lemma \ref{lem:pscentertoecc}, given a $(1+\eps)$-approximate algorithm $\sssp$ for SSSP running in $T(n,\eps,D)$ rounds, we may compute $(2+\eps^3 + 3\eps^2 + 4\eps)$-approximations for all eccentricities in $O(\log^2(n)\cdot T(n,\eps,D)+D)$ rounds.
\end{proof}

Using the $(1+\eps)$-approximate SSSP algorithms of \cite{becker_et_al:LIPIcs:2017:8003,DBLP:conf/focs/ForsterN18}, which run in $\tilde{O}((\sqrt{n} + D)/\eps)$ rounds on weighted, undirected graphs and $\tilde{O}((\sqrt{n}D^{1/4}+D)/\eps)$ rounds on weighted, directed graphs respectively, we achieve the following corollaries:

\begin{corollary-repeat}{corol:-eccentApprox1-intro}
\corolapproxone
\end{corollary-repeat}

\begin{corollary-repeat}{corol:-eccentapprox2-intro}
\corolapproxtwo
\end{corollary-repeat}

Using the exact SSSP algorithm of \cite{10.1145/3382734.3405729}, which runs in $\tilde{O}(\sqrt{n}D^{1/4}+D)$ rounds, we obtain the following corollary. 

\begin{corollary-repeat}{corol:-eccentapprox3-intro}
\corolapproxthree
\end{corollary-repeat}

\subsection{Approximations for unweighted undirected variants}
\label{subsec:ApxUwUd}

We now turn to our attention towards approximation algorithms for the unweighted, undirected versions of the distance parameters we consider.

\begin{theorem-repeat}{thm:approx_unweighted_undirected}
\approxunweightedundirected
\end{theorem-repeat}

We provide a distributed implementation of the sequential algorithm of Cairo et al. \cite{Cairo:2016:NBA:2884435.2884462}.
The high level approach for approximating the diameter is as follows. Suppose $d_1$ and $d_2$ are two nodes that realize the diameter. Our goal is to run BFS trees from various roots, and take the depth of the deepest tree as the estimate for the diameter. To obtain a better-than-2 approximation, we must promise that the set of roots includes a node that is sufficiently far from one of $d_1, d_2$ (for example, the depth of the BFS tree rooted at a node that lies exactly in between these two nodes approximates the diameter only by a factor of 2). The tension that we face is how to choose a small enough set of roots, since the size of this set directly affects the running time.

The approach for finding a node that is sufficiently far from $d_2$ is to find a node that is sufficiently close to $d_1$. Since we do not know which node is $d_1$ (otherwise we would simply run a BFS from it and find the exact diameter), we aim to find a set of nodes for which will be roots of BFS trees, such that every node has a close enough root (and in particular $d_1$ will have a sufficiently close root).

Cairo et al. \cite{Cairo:2016:NBA:2884435.2884462} introduce the following approach. An initial set of roots $S_0$ is sampled. Because of the need to keep this set small, there could be a node in the graph that is not close enough to the set. The algorithm finds the node $w_1$ that is farthest from this set, and would ideally like to add all nodes in its neighborhood of some sufficiently large radius into the set of roots, where the radius is such that it is guaranteed that all nodes in the graph are close to the new set of roots. However, this neighborhood could be too large. Thus, from this neighborhood, a smaller set $W_1$ of the closest nodes to $w_1$ is extracted. This set is still not promised to be sufficiently small, and hence the procedure is now repeated: a small enough subset $S_1$ is sampled from $W_1$ to be added to the set of roots, and again the farthest node from it is search for, in order to add more roots that promise that all nodes are close enough to the set of roots. This balancing act continues for $k$ iterations: adding more roots -- to be able to argue that all nodes have a close enough root, but not adding too many roots -- to be able to efficiently construct BFS trees from all of them.

\begin{proof}[Proof of Theorem \ref{thm:approx_unweighted_undirected}]
We show how to implement the algorithm of \cite{Cairo:2016:NBA:2884435.2884462} and refer there for correctness. 

We start by implementing a sampling procedure that outputs nodes $w_1,...,w_k$ and sets $S_0,...,S_k$.
Given $k$, let $W_0=V$ and let $\ell _0=n$, and let $q=\frac{n}{\log n}^{\frac{1}{k+1}}$.
First, the nodes compute a 2-approximation $D'$ of the diameter using a single BFS, and all nodes receive it. This is done in $O(D)$ rounds.

For each $i=0,...,k-1$, each node does the following. 
\begin{enumerate}

\item The nodes construct a set $S_i\subseteq V$  by letting each node in $W_i$ sample itself independently into the set with probability $\frac{q\cdot \log n}{n}$.
Define $Z_i=(V\backslash W_i)\cup S_i$.

\item Now, our goal is for each node $v$ to know the $id$ of the node farthest from $Z_i$, denoted by $w_{i+1}$. For this, the nodes of the network learn of $w_{i+1}$ in the following way. Each node knows whether it is in $Z_i$, so all nodes in $Z_i$ broadcast the message $(0, id)$, where $id$ is the unique ID of each broadcasting node. If a node $v$ in $V\backslash Z_i$ receives a set of messages of the form $(d,id)$ for the first time, it broadcasts $(d_{min}+1,id(v))$ in the next round, where $d_{min}$ is the smallest value of $d$ it received. After broadcasting their first message, all nodes thereafter broadcast the largest $(d,id)$ message they have received. After $D'$ rounds, all nodes will have the same value $(d_{max},id)$ stored, so whichever node has UID $id$ knows it is $w_{i+1}$. This routine takes $O(D)$ time.

\item Next, the nodes compute a BFS tree $T$ rooted at $w_{i+1}$, and compute $W_{i+1}$, which is the set of the first $\ell_{i+1} =\ceil{\frac{\ell_i}{q}}$ nodes encountered. This set can be learned in the following way. First $w_{i+1}$ can learn the values $|N_j(w_{i+1})|$ for all $j$, which is the number of nodes within distance at most $j$ from $w_{i+1}$. This is done using the BFS tree $T$ and an aggregate summation of the $T$-degrees of all nodes. Now, $w_{i+1}$ finds a value $j$ such that $$|N_j(w_{i+1})|\leq \ell_{i+1}  \leq |N_{j+1}(w_{i+1})|$$

Next, using another aggregate computation on the tree, each node in $N_j(w_{i+1})$ can learn how many nodes of distance $j+1$ there are in its subtree. Then, $w_{i+1}$ picks sufficiently many descendants in the tree such that they can choose $|N_{j+1}(w_{i+1})|-\ell_{i+1}$ nodes into the set $W_{i+1}$ from their own descendants, and so on, where at the end, ties are being broken arbitrarily.
This process takes $O(D)$ rounds as well. At the end of this process, each node $v$ knows whether $v\in W_{i+1}$ or not. If it holds that $W_{i+1}\cap Z_i\neq \emptyset$, then we continue to the $(i+1)$-th iteration. This can be checked since all nodes know whether they are in $Z_i$ and $W_{i+1}$, and if a node is in both sets, it can broadcast that information in $O(D)$ rounds using an aggregate computation. If the check fails, the $i$-th iteration is repeated.

\end{enumerate}
~\\This completes the description of the loop. Finally, at the end of the last iteration, we set $S_k=W_k$.
Note that each iteration takes $O(D)$ rounds, and it is proven in \cite{Cairo:2016:NBA:2884435.2884462} that the condition of $W_{i+1}\cap Z_i\neq \emptyset$ at the end of each iteration holds w.h.p.

Once we have obtained $w_1,...,w_k$, and the sets $S_0,...,S_k$ we 
run a BFS from each $w_i$ and all elements of each $S_i$. This gives that these nodes know their eccentricities and all nodes know their distance to each $w_i$ and all elements of $S_i$, for all $i$.
These $\tilde{O}(n^{\frac{1}{k+1}})$ BFS computations complete in $\tilde{O}(n^{\frac{1}{k+1}}+D)$ rounds using the algorithm of \cite{LenzenP13}, which allows computing a BFS from every node in a given set $S$ in an unweighted, undirected graph can be done in $O(S+D)$ rounds. 

Now, each node estimates its eccentricity by the farthest node from itself it is aware of: If the node $v$ is one of $S_0,...,S_k$, or one of $w_1,...,w_k$, then $v$ knows its distance to all nodes in the graph. Otherwise, $v$ is only aware of its distances to the nodes in $S_0,...,S_k$, and the nodes $w_1,...,w_k$. The estimated diameter is computed using an aggregation of the maximal estimated eccentricity, and the radius is estimated using an aggregation of the minimal estimated eccentricity. Both aggregate computations take $O(D)$ rounds.
\end{proof}

\subsection{Approximations for ST variants}
\label{subsec:ApxST}

Lastly, we prove Theorems \ref{thm:approx_bichro_diameter_unweighted_intro} and \ref{thm:approx_bichro_diameter_weighted_intro}.
The proof of Theorem \ref{thm:approx_bichro_diameter_unweighted_intro} follows the same ideas as the proof of Theorem \ref{thm:approx_unweighted_undirected}, with the appropriate modifications for the bi-chromatic variant.
\begin{theorem-repeat}{thm:approx_bichro_diameter_unweighted_intro}
\approxupperbichromaticdiameter
\end{theorem-repeat}

\begin{proof}
We give a distributed implementation of the sequential algorithm given by Dalirrooyfard et al. in \cite{DBLP:conf/icalp/DalirrooyfardW019a}.
We denote by $s^*,t^*$ the nodes that satisfy $d(s^*,t^*)=D_{ST}$.

Throughout the algorithm, there are going to be 5 different estimates of the diameter, $D_1,\dots,D_5$, and then we output the maximal one among them as our approximation for the bi-chromatic diameter. The outline of the proof is that we run BFS trees from sets of nodes, such that for one of these nodes, $v$, its distance to either $s^*$ or $t^*$ is at least $3D_{ST}/5$, which gives the claimed approximation (the small additive term comes from rounding in case $D_{ST}/5$ is not an integer).

\textbf{Estimate $D_1$:}
    Each node $v\in S$ samples itself independently into a set $Z$ with probability  $\frac{c\sqrt{n}\ln n}{|S|}$ for some constant $c$, and each node $u\in T$ samples itself independently into a set $X$ with probability $\frac{C\sqrt{n}\ln n}{|T|}$, for some sufficiently large constant $C$. By standard arguments, the sizes of the sets $Z$ and $X$ are $\tilde{O}(\sqrt{n})$, w.h.p.
The algorithm then runs BFS trees from all nodes in $Z,X$, which can be done in $O(|Z|+|X|+D)$ rounds by \cite{LenzenP13}, and thus completes in $\tilde{O}(\sqrt{n}+D)$ rounds w.h.p. Let $D_1=\max_{z\in Z,t\in T} d(z,t)$. The value of $D_1$ can be made known to all nodes in $O(D)$ rounds using a standard aggregate maximum computation. 
Note that if for some $z\in Z$ it holds that $d(s^*,z)\leq \frac{2D_{ST}}{5}$, then $d(z,t^*)\geq \frac{3D_{ST}}{5}$. Since $D_1\geq d(z,t^*)$, in this case our output at the end is within the claimed approximation.

\textbf{Estimate $D_2$:}
For every $t\in X$, denote by $s(t)$ the closest node in $S$ to $t$, which can be computed by a simple aggregation on the BFS tree from $t$, which for all trees is bounded by $O(|X|+D)$ rounds.
We now run BFS trees from all nodes $\set{s(t)\mid t\in X}$. Denote $D_2=\max\limits_{t\in X,t'\in T} d(s(t),t')$. 
Note that if for some $t\in X$, it holds that $d(s^*,t)\leq \frac{D_{ST}}{5}$, then $d(s^*,s(t))\leq \frac{2D_{ST}}{5}$ since $s(t)$ is closer to $t$ than $s^*$. Thus we get that $d(s(t),t^*)\geq \frac{3D_{ST}}{5}$. Since $D_2\geq d(s(t),t^*)$, in this case we also output a good estimate.

~\\\textbf{Estimate $D_3,D_4,D_5$:}
If $D_1,D_2$ do not achieve the needed approximation, this means that $d(s^*,X)>\frac{D_{ST}}{5}$, and $d(s^*,Z)>\frac{2D_{ST}}{5}$. Our next goal is to find such a node $w$, that satisfies $d(w,X)>\frac{D_{ST}}{5}$, and $d(w,Z)>\frac{2D_{ST}}{5}$.

For each $s\in S$, denote by $D_s$ the largest integer for which $d(s,X)>\frac{D_s}{5}$ and $d(s,Z)>\frac{2D_s}{5}$. Each node can compute this integer internally, and the network can compute $w=\arg \max\limits_{s\in S} D_s$ and $D'=\max\limits_{s\in S} D_s$ in $O(D)$ rounds. 
Note that $D'\geq D_{ST}$ since $d(s^*,X)>\frac{D_{ST}}{5}$, and $d(s^*,Z)>\frac{2D_{ST}}{5}$.

As proven in \cite{DBLP:conf/icalp/DalirrooyfardW019a}, it suffices to look at the nodes of $S$ of distance at most $\frac{2D'}{5}$ from $w$, from which we get the estimates $D_3,D_4$, and nodes of $T$ within distance at most $\frac{D'}{5}$ of $w$, from which we get the estimate $D_5$, as we explain next.

\textbf{Estimate $D_3,D_4$:}
We now run a BFS tree from $w$, and denote by $S_w$ all nodes of $S$ at distance at most $\frac{2D'}{5}$ from $w$. It is proved in \cite{DBLP:conf/icalp/DalirrooyfardW019a} that w.h.p. $|S_w|\leq \sqrt{n}$. Thus we can run BFS trees from all nodes in $S_w$ in $O(\sqrt{n}+D)$ rounds, and we denote $D_3=\max\limits_{s\in S_w,t\in T} d(s,t)$. 
For every node $s\in S_w$, we denote by $t(s)$ the closest node in $T$ to $s$. We now run BFS trees from all nodes $\set{t(s)\mid s\in S_w}$, and denote $D_4=\max\limits_{s\in S_w,s'\in S} d(t(s),s')$. These values can again be learned by the entire network using an aggregate maximum computation.

Here too, if we get that $D_3\geq \frac{3D_{ST}}{5}$ or $D_4\geq \frac{3D_{ST}}{5}$, then again our estimate is a good approximation.
Otherwise if both $D_3,D_4$ are smaller than $\frac{3D_{ST}}{5}$, since $D_3\geq d(w,t^*)$, in particular $d(w,t^*)<\frac{3D_{ST}}{5}$. Now we look at the node $b$ on the shortest $P_{wt^*}$ from $w$ to $t^*$ that satisfies $d(w,b)=\frac{2D_{ST}}{5}$. 
The proof then shows, using the information known thus far and triangle inequalities, that it must hold that $b\in T$, since otherwise $D_4\geq \frac{3D_{ST}}{5}$.

\textbf{Estimate $D_5$:}
Finally, denote by $T_w$ the set of nodes of $T$ at distance at most $\frac{D'}{5}$ from $w$. It is proved in \cite{DBLP:conf/icalp/DalirrooyfardW019a} that w.h.p. $|T_w|\leq \sqrt{n}$. In another $O(\sqrt{n}+D)$ rounds, we run BFS trees from all nodes in $T_w$, and denote $D_5=\max\limits_{t\in T_w,s\in S} d(s,t)$.

To conclude the proof, assume towards a contradiction that $D_5<\frac{3D_{ST}}{5}$, and consider the node $a$ on the shortest $P_{wt^*}$ from $w$ to $t^*$ that satisfies $d(w,b)=\frac{D_{ST}}{5}$. Similarly to before, the proof shows that it must hold that $a\in S$.
Now, the proof shows that in this case, it must hold that $D_4$ is in fact a good approximation.
If $a\in S$ and $b\in T$, there has to be an edge $(s',t'),s'\in S,t'\in T$ on the path $P_{t^*w}$ on the part between $a,b$. In particular this means that $s'\in S_w$, so we ran BFS from $t(s')$. The proof shows, using triangle inequalities, that $d(s^*,t(s))\geq \frac{3D_{ST}}{5}$, and since $D_4\geq d(s^*,t(s))$, the analysis is concluded. 

\end{proof}

\begin{theorem-repeat}{thm:approx_bichro_diameter_weighted_intro}
\approxweightedupperbichromaticdiameter 
\end{theorem-repeat}

\begin{proof}
We implement the sequential algorithm of Dalirrooyfard et al. \cite{DBLP:conf/icalp/DalirrooyfardW019a} in the distributed setting. 
The approach of the algorithm is to find the minimum weight edge between $S$ and $T$, denoted by $(s,t)$, and to let $D'=\max\{\max_{s'\in S} d(s',t), \max_{t'\in T}d(s,t')\}$ be our estimate of the $ST$ diameter. Using the triangle inequality, this can be shown to provide an approximate solution as stated.

To find the edge $(s,t)$, each vertex broadcasts its minimum weight $ST$ edge in the first step, and then upon receiving messages from neighbors, each node updates and broadcasts the minimum weight known to it. After $O(D)$ rounds all nodes have the minimum $ST$ edge, where $D$ is the diameter of the graph. 

Now, the nodes run two SSSP instances, one from $s$ and one from $t$, in $O(T(SSSP))$ rounds. In another $O(D)$ rounds of propagating the maximal distances, all nodes can compute the required estimate $D'=\max\{\max_{s'\in S} d(s',t), \max_{t'\in T}d(s,t')\}$. Since $T(SSSP)\ge O(D)$, the total number of rounds is $O(T(SSSP))$. 
\end{proof}


%% file: lowerbounds.tex
In this section, we prove the lower bound results of the paper. As stated, we use reductions from 2-party communication complexity.
To formalize the reductions, we restate the following definition from Censor-Hillel et al. \cite{DBLP:conf/wdag/Censor-HillelKP17}.

\begin{definition}[Family of Lower Bound Graphs]
\label{def: family of lb graphs}
	Given integers $K$ and $n$, a Boolean function $f:\{0,1\}^{K} \times \{0,1\}^{K} \to \{0 ,1 \}$ and some Boolean graph property or predicate denoted $P$,
	a set of graphs $\set{G_{x,y}=(V,E_{x,y}) \mid x,y\in \{0,1\}^K}$ is called a \emph{family of lower bound graphs} with respect to $f$ and $P$ if the following hold:
	\begin{enumerate}
		\item The set of vertices $V$ is the same for all the graphs in the family, and we denote by $V_A,V_B$ a fixed partition of the vertices.
		\item Given $x,y\in \{0,1\}^{K}$, the only part of the graph which is allowed to be dependent on $x$ (by adding edges or weights, no adding vertices) is $G[V_A]$.
		\item Given $x,y\in \{0,1\}^{K}$, the only part of the graph which is allowed to be dependent on $y$ (by adding edges or weights, no adding vertices) is $G[V_B]$.
		\item $G_{x,y}$ satisfies $P$ if and only if $f(x,y)=1$.
	\end{enumerate}
    The set of edges $E(V_A,V_B)$ is denoted by $E_{\cut}$, and is the same for all graphs in the family.
\end{definition}
We use the following theorem whose proof can be found in Censor-Hillel et al. \cite{DBLP:conf/wdag/Censor-HillelKP17}, with $CC^R(f)$ denoting the randomized communication complexity of $f$.

\begin{theorem}\label{generallowerboundtheorem}
	Fix a function $f:\{0,1\}^{K} \times \{0,1\}^{K} \to \{0 ,1 \}$ and a predicate $P$. If there exists a family of lower bound graphs $\{G_{x,y} \}$ w.r.t $f$ and $P$, then
	every randomized algorithm for deciding $P$ takes $\Omega(CC^R (f)/(\size{E_{\cut}} \log n))$ rounds.
\end{theorem}

\subsection{Lower bounds for radius}
\label{ssec:-radius_hard}
We start with proving our two lower bounds for weighted or directed radius approximations.

We divide the proof of \cref{thrmintro:-weightedDirectedRadiusLowerBoundTwo} into two cases which we prove separately.

\begin{theorem-repeat}{thrmintro:-weightedDirectedRadiusLowerBoundTwo}[\textbf{Weighted case}]
For any $\eps=1/\poly(n)$, $(2-\eps)$-approximation of the radius of a weighted graph with $n$ nodes requires $\Omega(n/\log n)$ rounds, even when the graph has constant hop-diameter.
\end{theorem-repeat}
\begin{proof}
We reduce from the Tribes problem with vector sets $A$ and $B$ of size $N$. This construction is similar to that of \cite[Theorem 7]{holzer_et_al:LIPIcs:2016:6597}.

\begin{figure}[ht]
        \begin{center}
        \includegraphics[width=.7\textwidth]{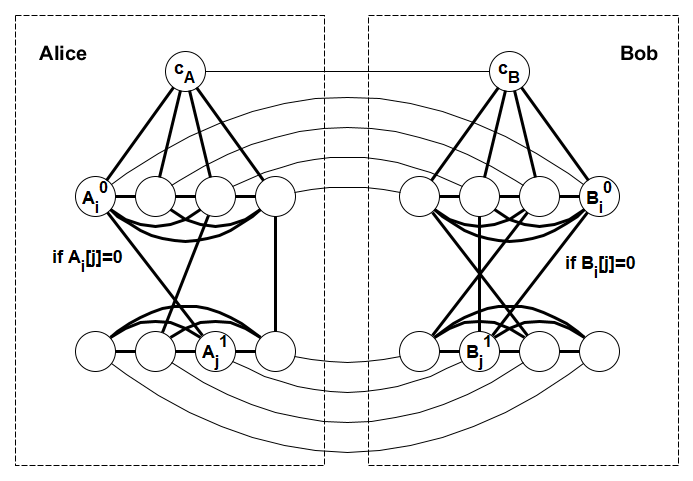}
        
        \caption{Sketch of Theorem \ref{thrmintro:-weightedDirectedRadiusLowerBoundTwo}, weighted case construction. Bold lines represent edges of weight $t$.}
        \label{fig:-wtedRad2}
        \end{center}
\end{figure}

Figure \ref{fig:-wtedRad2} illustrates our family of lower bound graphs. We construct four cliques $A^0,A^1,B^0,B^1$ of size $N$, where the edges of the cliques have weight $t$, a value we will set later. Let $K_i$ be the $i$th node in clique $K$. Add two nodes $c_A$ and $c_B$.

Connect all nodes in $A^0$ to $c_A$ with edges of weight $t$, and connect all nodes in $B^0$ to $c_B$ with edges of weight $t$. Connect $c_A$ and $c_B$ with an edge of weight 1. For all $i\in[N]$ and $b\in\set{0,1}$, connect $A_i^b$ and $B_i^b$ with an edge of weight 1. Connect $A_i^0$ and $A_j^1$ with an edge of weight $t$ if and only if $A_i[j]=0$. Connect $B_i^0$ and $B_j^1$ with an edge of weight $t$ if and only if $B_i[j]=0$. Alice will simulate the nodes $A^0\cup A^1\cup\set{c_A}$, and Bob will simulate the nodes $B^0\cup B^1\cup\set{c_B}$.

First, we claim that if $(A,B)$ is a `yes' instance of Tribes, then the radius is at most $t+2$. To show this, note that in this case, there must be some $i$ such that the $i$th vectors of $A$ and $B$ are orthogonal. Consider the node $A_i^0$. It may reach in distance at most $t+1$ all nodes in $B^0\cup A^0$, via a clique edge and an edge in the matching between $A^0$ and $B^0$. It may also reach $\set{c_A,c_B}$ in at most $t+1$. It may also reach all nodes in $A^1\cup B^1$ in distance at most $t+2$, because for any $j$ where $A_i[j] = 0$ or $B_i[j]=0$, either $A_i^0$ may reach $A_j^1$ in distance $t$ or $B_i^0$ may reach $B_j^1$ in distance $t$. Since $A_i$ and $B_i$ are orthogonal, this is true for all $j$. Thus the eccentricity of $A_i^0$ is at most $t+2$, which upper-bounds the radius.

Second, we claim that if $(A,B)$ is a `no' instance of Tribes, then the radius is at least $2t$. To see this, first note that $c_A$ and $c_B$ have eccentricity at least $2t$, because that is the shortest possible distance between them and $B^1\cup A^1$. By the same argument, the eccentricity of all nodes in $A^1\cup B^1$ is also at least $2t$. For all $i$, $A_i$ and $B_i$ are not orthogonal, which means that for all $i$ there is some $j$ such that neither $A_i^0$ nor $B_i^0$ has an edge to $B_j^1$ or $A_j^1$. Clearly any other path from $B_i^0$ or $A_i^0$ to $B_j^1$ or $A_j^1$ is at least of length $2t$, via a clique edge of weight $t$. Thus the eccentricities of all nodes are at least $2t$, so the radius is at least $2t$.

We set $t=\ceil{\frac{4}{\eps}}$ so that a $(2-\eps)$-approximate radius algorithm needs to distinguish between $t+2$ and $2t$. The constructed graph $G_{A,B}$ has $n=O(N)$ nodes with a cut of size $O(n)$, which by Theorem \ref{generallowerboundtheorem} and the lower bound of $\Omega(N^2)$ for the communication complexity of Tribes, implies that the radius algorithm requires $\Omega(n/\log{n})$ rounds.
\end{proof}


\begin{theorem-repeat}{thrmintro:-weightedDirectedRadiusLowerBoundTwo}[\textbf{Directed case}]
For any $\eps>0$, $(2-\eps)$-approximation of the radius of a directed graph with $n$ nodes requires $\Omega(n\eps/\log^2(n\eps))$ rounds, even when the graph has constant hop-diameter and $\tilde{O}(n\eps)$ edges.
\end{theorem-repeat}
\begin{proof}
We reduce from the HSE problem on sets $A$ and $B$ of size $N$ and vectors of size $d=2\log{N}+1$. First, if Alice detects a coordinate that is 0 for all elements of $A$, she conveys this information to Bob, and both remove that coordinate from all vectors. This requires only $O(\log{N})$ bits of communication.

\begin{figure}[ht]
        \begin{center}
        \includegraphics[width=\textwidth]{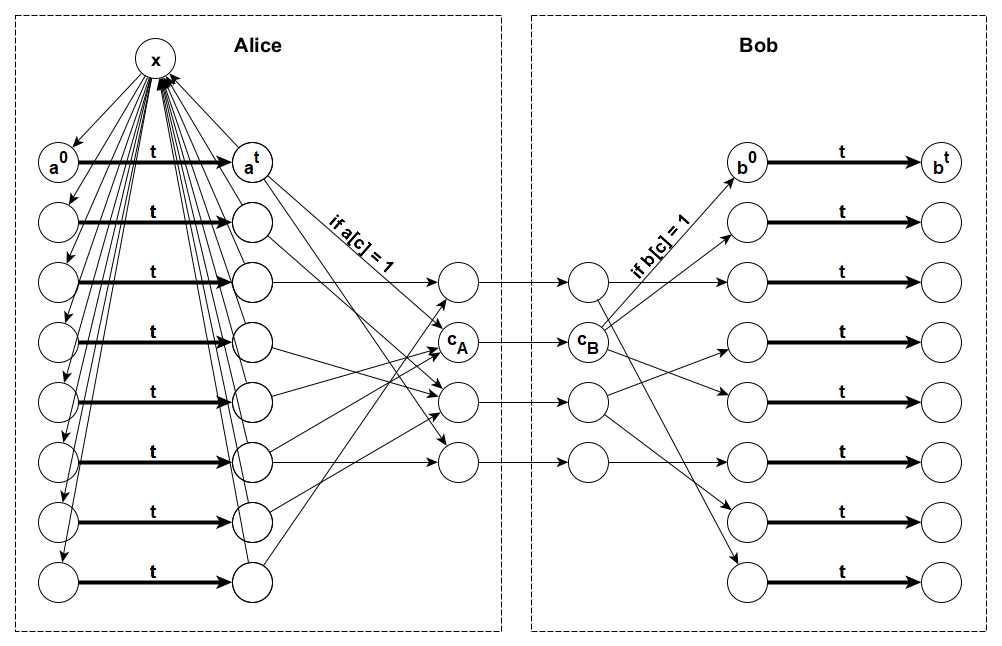}
        \caption{Sketch of Theorem \ref{thrmintro:-weightedDirectedRadiusLowerBoundTwo}, directed case construction. Bold lines represent paths of length $t$.}
        \label{fig:-dirRad2}
        \end{center}
\end{figure}

Figure \ref{fig:-dirRad2} illustrates our family of lower bound graphs, inspired by a construction from the sequential setting that reduces HSE to the source radius problem \cite[Theorem 1.4]{Abboud:2016:AFP:2884435.2884463}.
Let $t$ be an integer to be decided later. For each $a\in A$, create a directed path of length $t$; denote the first node in the path by $a^0$ and the last by $a^t$. Do the same for each $b\in B$. Add a node $x$, and for each $a$, add edges $(a^t,x)$ and $(x,a^0)$. Next, for each coordinate $c\in[d]$, create two nodes $c_A$ and $c_B$, with the edge $(c_A,c_B)$. For each $a\in A$, if $a[c]=1$, add edge $(a^t,c_A)$. For each $b\in B$, if $b[c]=1$, add edge $(c_B,b^0)$. Alice will simulate the $a$-paths, $x$, and the nodes $c_A$, and Bob will simulate the rest.

First, we claim that if there exists a hitting set $h\in A$, then the radius of the graph is at most $t+4$. To see this, consider the node $h^t$. It may reach all nodes $a^t$ in distance at most $t+2$ via $x$, all nodes $c_A$ in distance at most $t+3$ via the nodes $a^t$, and all nodes $c_B$ in distance at most $t+4$. There are no coordinates that are 0 for all $a\in A$, so each $c_A$ is reachable from {\em some} node $a^t$. Finally, note that because $h$ is a hitting set, for each $b\in B$ there is some coordinate $c$ such that $h[c]=b[c]=1$, so there exists the path $(h^t,c_A,c_B,b^0)$. From there, $h^t$ may reach all nodes on the $b$-path in at most $t$ additional steps. Thus, the eccentricity of $h^t$ is at most $t+4$, which upper-bounds the radius.

Second, we claim that if there is no hitting set, the radius is at least $2t+4$. For this, note that the only candidates for the center of the graph are the nodes $a^t$ and the node $x$, because all nodes on $a$-paths have eccentricity greater than the corresponding $a^t$, and all other nodes have infinite eccentricity. The node $x$ is clearly at distance at least $2t+4$ from every node $b^t$. Fix a vector $a$. There is some vector $b\in B$ such that $a\cdot b=0$; fix the value of $b$. There is no path of length $3$ from $a^t$ to $b^0$, because there is no coordinate that is 1 for both vectors. The only way for $a^t$ to reach $b^0$ is via $x$, at an additional cost of $t+2$ distance, so the distance from $a^t$ to $b^t$ is at least $2t+5$. Thus, every node has eccentricity at least $2t+4$, which lower-bounds the radius.

If we set $t$ such that $\frac{2t+4}{t+4}>2-\eps$, any algorithm for $(2-\eps)$-approximate radius must distinguish between the two cases. Note that the graph $G_{A,B}$ has $n = O(Nt) = O(N/\eps)$ nodes and $O(N\log{N})$ edges, with a cut of size $O(\log{N}) = O(\log(n\eps)) = O(\log{n})$. Thus, by Theorem \ref{generallowerboundtheorem} and the lower bound of $\Omega(N)$ for the communication complexity of HSE that we proved in Theorem \ref{thrm:-HSEhard}, any algorithm for $(2-\eps)$-approximate radius requires $\Omega(n/\log^2{n})$ rounds of communication.
\end{proof}


\subsection{Lower bounds for bi-chromatic diameter}\label{ssec:-bichromaticD_hard}

Now, we turn to prove our two lower bounds for bi-chromatic diameter.

\begin{theorem-repeat}{thm:lowerbound_diameter_bichromatic_intro}
\approxweightedlowerbichromaticdiameter
\end{theorem-repeat}

\begin{proof}
We reduce the OV problem on sets $A$ and $B$ of size $N$ and vectors of size $d=2\log{N}+1$ to an instance of bichromatic diameter on a undirected graph $G$. Without loss of generality we can assume that there is a coordinate $\hat{c}\in [d]$ such that $b[\hat{c}]=1$ for all $b\in B$ and $a[\hat{c}]=0$ for all $a\in A$. If no such coordinate exist, we can add it without affecting the OV instance. Let $t$ be an integer to be decided later. For each $a\in A$, let $a^0,\ldots,a^t$ be a path of length $t$ in $G$. For each $b\in B$, add a single node $b^0$. For each coordinate $c\in [d]$, create two nodes $c_A$ and $c_B$, with an edge $c_Ac_B$. For each $a\in A$, if $a[c]=1$, add a path of length $t$ from $a^t$ to $c_A$. Similarly, for each $b\in B$, if $b[c]=1$, add a path of length $t$ from $c_B$ to $b^0$. Let $S$ be the union of the set of all $a_i$ nodes, all $c_A$ nodes and all $c_B$ nodes. Let $T$ be the rest, i.e., $B=\{b^0~|~b\in B\}$. This completes the definition of the graph. 

Alice will simulate the subgraph of $a$ nodes and $c_A$s, and Bob will simulate the rest, which is the subgraph on $b$ nodes and $c_B$s. Note that the set Alice simulates is not $S$, it is a subset of $S$. 

First, suppose that there is no orthogonal pair in the OV instance. Then for each $a\in A$ and $b\in B$, there is a coordinate $c$ where $a[c]=b[c]=1$. So the path from $a^0$ to $b^0$ going through $a^t, c_A$ and $c_B$ has distance $3t+1$. So the distance between any $a^i$ to $b^0$ is at most $3t+1$. The distance between any coordinate node $c_B$ and any $b^0$ is at most $3$, since for some $b'$ where $b'[c]=1$, the path $c_B b'^0 \hat{c}_B b^0$ connects the two nodes. So the bichromatic diameter is $3t+1$ in this case.

Now, suppose that there is an orthogonal pair $(a,b)$ in the OV instance. Since there is no path of the form $a^0,\ldots,a^t,c^A,c^B,b^0$ for any coordinate $c$, the distance between $a^0$ and $b^0$ is at least $5t+1$, as a path between $a^0$ and $b^0$ must either first pass a node $b'^0$, or pass a node $a'^t$ through a coordinate $c_A$. So the bichromatic diameter is at least $5t+1$ in this case.

If we set $t$ such that $\frac{5t+1}{3t+1}>\frac{5}{3}-\epsilon$, any algorithm for $(\frac{5}{3}-\epsilon)$-approximate bichromatic diameter must distinguish the two cases. The graph has $n=O(tN\log{N})=O(N\log{N}/\eps)$ nodes and edges, with a cut of size $O(\log{N})=O(\log{n\epsilon})=O(\log{n})$. So, by Theorem \ref{generallowerboundtheorem} and the lower bound of $\Omega(N)$ on the communication complexity of OV, any algorithm for $(5/3-\epsilon)$-approximation bichromatic radius requires $\Omega(n/\log^3{n})$ rounds of communication.
\end{proof}

\begin{theorem-repeat}{thm:lowerbound_diameter_bichromatic_directed_intro}
\twoapproxbichromaticdirecteddiameter
\end{theorem-repeat}

\begin{proof}
We modify our construction for the proof of Theorem \ref{thm:lowerbound_diameter_bichromatic_intro}. We reduce the OV problem on sets $A$ and $B$ and vectors of size $d=2\log{N}+1$ to an instance of bichromatic diameter on a directed graph. First, recall that for each $a\in A$, we add a path $a^0,\ldots,a^t$ of length $t$, where we direct the edges from $a^i$ to $a^{i+1}$ for all $i=0\ldots,t-2$.
For each $b\in B$, we add a node $b^0$ and for each coordinate $c\in [d]$, we add two nodes $c_A$ and $c_B$, where there is an edge from $c_A$ to $c_B$. For each $a\in A$ and $c\in [d]$, connect $a^t$ to $c_A$ if $a[c]=1$. For each $b\in B$ and $c\in [d]$ connect $c_B$ to $b^0$ if $b[c]=1$. Let $P=p_0,\ldots,p_{t-2}$ be a directed path of length $t-2$. Connect all nodes $b^0$ to $p_0$, and connect $p_{t-2}$ to all the coordinate nodes $c_B$. Let $T=\{b^0~|~b\in B\}$ and let $S$ be the rest of the graph. This completes the definition of the bichromatic instance. 

Bob will simulate the subgraph induced on $T$, the path $P$ and the coordinate nodes $c_B$. Alice will simulate the subgraph induced on the rest of the nodes. 

First, suppose that there is no orthogonal pair in the OV instance. Thus, for each $a\in A$ and $b\in B$, there is a coordinate $c$ such that $a[c]=b[c]=1$, and hence $d(a^0,b^0)=t+3$. For any $c\in [d]$ and $b\in B$, we have that $d(c_A,b^0)\le t+3$. This is because if $b[c']=1$ for some coordinate $c'\in [d]$, then the path $c_A,c_B,b'^0,p_0,\ldots,p_{t-2},c'_B,b^0$ connects $c_A$ and $b^0$ for some $b'\in B$ where $b'[c]=1$. So the $ST$ diameter is $t+3$ in this case. 

Now, suppose that there is an orthogonal pair $(a,b)$ in the OV instance. In this case, the path from $a^0$ to $b^0$ has to use $P$, since there is no path from $a^0$ to $b^0$ in $G\setminus P$. Hence, we have that $d(a^0,b^0)\ge 2t+3$. If we set $t$ such that $\frac{2t+3}{t+3}\ge 2-\epsilon$, any algorithm for $(2-\epsilon)$-approximate bichromatic diameter must distinguish the two cases. The graph has $n=O(tN)=O(N/\epsilon)$ nodes and $O(N\log{N}/\epsilon)$ edges, with a cut of size $O(\log{N})=O(\log{n})$. So by Theorem \ref{generallowerboundtheorem} and the lower bound of $\Omega(N)$ on the communication complexity of OV, any algorithm for $(2-\epsilon)$-approximation bichromatic diameter in directed graphs requires $\Omega(n/\log^2{n})$ rounds of communication.
\end{proof}

\subsection{Lower bounds for arbitrary approximation ratios}\label{ssec:any_approx_lower_bound}
Our main approach in order to show lower bounds for arbitrary approximation ratios is to reduce from the SCSV problem defined next. A challenge is to make sure that the reductions one employs can be efficiently simulated in \cgst. A standard combinatorial reduction does not suffice since at the end of the day, the vertices of the graph need to be able to simulate the reduction in order to solve the original problem. We overcome this by coming up with suitable reductions in which every round in the new graph can be simulated using a constant amount of rounds on the original graph. 
\begin{definition}[The Spanning Connected Subgraph Verification (SCSV) Problem {\cite[Section 4]{DasSarma:2011:DVH:1993636.1993686}}]\label{def:-SCSV}
Given a connected, unweighted, undirected graph $G$, in which a subset of edges are marked as being part of a subgraph $H$ of $G$, determine whether $H$ is a connected, spanning subgraph of $G$.
\end{definition}
It is shown by Das Sarma et al. in \cite[Theorem 5.1]{DasSarma:2011:DVH:1993636.1993686} that solving SCSV requires $\tilde{\Omega}(\sqrt{n}+D)$ rounds. Even for randomized algorithms succeeding with high probability.

\begin{theorem-repeat}{thrmintro:-weightedDiameterLowerBoundAny}
\anyapproxweighteddiameter
\end{theorem-repeat}
\begin{proof}
We reduce from SCSV to any approximation of weighted diameter, radius, or all eccentricities. Let $\mathcal{A}$ be an algorithm for one of these problems with approximation ratio $\alpha(n)$.  We begin with an instance of SCSV with graph $G=(V,E)$ and subgraph $H = (V, E_H)$. We set the weight of all edges in the subgraph to 1, and the weight of all other edges to $n\cdot\alpha(n)$. This requires no additional communication. Note that if the subgraph is spanning and connected, all eccentricities are at most $n-1$ using edges of the subgraph of weight 1. Otherwise, the graph is broken up into at least two connected components of the subgraph connected by edges of weight $n\cdot\alpha(n)$, and the eccentricity of every node is thus at least $n\cdot\alpha(n)$. $\mathcal{A}$ must distinguish between these two cases, since it approximates the maximum (diameter) eccentricity. If $\mathcal{A}$ is for bi-chromatic diameter, we may set any node as the sole element of $S$ and all others to $T$, and detect whether the bi-chromatic diameter is at most $n-1$ or at least $n\cdot\alpha(n)$. Applying the lower bound of $\tilde{\Omega}(\sqrt{n}+D)$ rounds for SCSV gives that $\mathcal{A}$ requires $\tilde{\Omega}(\sqrt{n}+D)$ rounds.
\end{proof}

\begin{theorem-repeat}{thm:bichromatic_directed_diameter_lowerbound_any}
\anyapproxbichromaticdirecteddiameter
\end{theorem-repeat}
\begin{proof}
We reduce from SCSV to any approximation of directed bi-chromatic diameter. We begin with an instance of SCSV with graph $G=(V,E)$ and subgraph $H = (V, E_H)$. We define a new graph $G'=(V',E')$, as follows, essentially duplicating $H$ and directing it to $G$ (see Figure \ref{figure:dirRadAny}). For each $v\in V$, there are two nodes $v_G$ and $v_H$ in $V'$, and $(v_H,v_G)\in E'$. If $(u,v)\in E$, then $(u_G,v_G)\in E'$ and $(v_G,u_G)\in E'$, and if $(u,v)\in E_H$, then $(u_H,v_H),(v_H,u_H)\in E'$. Recall that $G$ is connected, and every $v_H$ has an edge to $v_G$, so every $v_H$ can reach every $u_G$ in $V'$. Also, every $v_G$ has infinite eccentricity, as there is no $(u_G,v_H)$ edge in $V'$.

\begin{figure}[ht]
        \begin{center}
        \includegraphics[width=.7\textwidth]{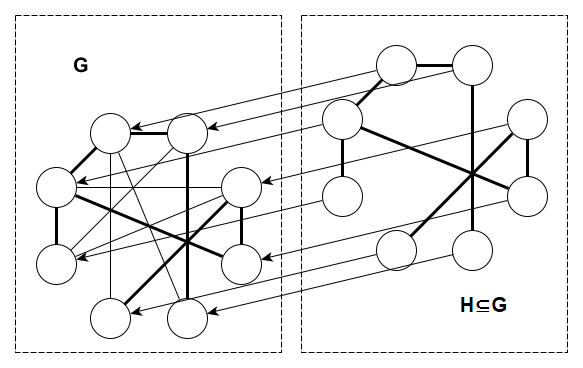}
        \caption{Sketch of Theorem \ref{thm:bichromatic_directed_diameter_lowerbound_any} construction. Lines without arrows denote bidrectional edges, and bold lines highlight subgraph $H$ in $G$.}
        \label{figure:dirRadAny}
        \end{center}
\end{figure}

We claim that $H$ is spanning and connected iff every $v_H\in V'$ has finite eccentricity. If $H$ is spanning and connected, there is a path from each $v_H$ to all other $u_H$, and therefore all $v_H$ can reach all nodes in the graph in distance at most $n$. Otherwise, $H$ has at least two connected components, and every node $v_H$ cannot reach any node in a different component from itself. Thus, all $v_H$ have infinite eccentricity, which means that all nodes in $G'$ have infinite eccentricity and the radius is infinite. 

Given any approximation algorithm for $\mathcal{A}$ for bi-chromatic diameter, we may set $S=\set{v_H}$ for any arbitrary member of $v_H\in V'$ and set $T$ to all other nodes, and the bi-chromatic diameter will be the eccentricity of $v_H$, either finite or infinite.

We complete our reduction by pointing out that a node $v$ in the original instance may simulate $v_H$ and $v_G$, at a cost of doubling the number of rounds of the approximation algorithm. This is because the original edges are doubled in the reduction, and the only edges added in the reduction are between $v_H$ and $v_G$, which are simulated by the same node. No additional communication is required for the reduction, so we apply the lower bound of $\tilde{\Omega}(\sqrt{n}+D)$ rounds for SCSV and achieve the same lower bound of $\tilde{\Omega}(\sqrt{n}+D)$ for any approximation algorithm $\mathcal{A}$ for radius or eccentricities.
\end{proof}

\begin{theorem-repeat}{thrmintro:-unweightedDirectedDiameterLowerBoundAny}
\anyapproxdirecteddiameter
\end{theorem-repeat}
\begin{proof}
We reduce from SCSV to any approximation of directed diameter. We use our construction for the proof of Theorem \ref{thm:bichromatic_directed_diameter_lowerbound_any}, with the following change: for an arbitrary node $v\in G$, we add the edge $(v_G,v_H)$ in $G'$. Note that for all $v\neq v'\in G$ we have that $(v'_G,v'_H)\notin E'$.

If the SCSV instance gives a spanning, connected subgraph, then all nodes may reach all other nodes in $G'$, by similar arguments as in the proof of Theorem \ref{thm:bichromatic_directed_diameter_lowerbound_any}; therefore, the eccentricity of all nodes, and therefore the diameter, is finite. However, if $H$ is not spanning, then consider a node $v'\in G$ that is not reachable from $v$ via edges in $H$. There is no path from $v_G$ to $v'_G$ in $G'$ either, so the diameter is infinite in this case.
Since no additional communication is required for the reduction, we apply the lower bound of $\tilde{\Omega}(\sqrt{n}+D)$ rounds for SCSV and achieve the same lower bound of $\tilde{\Omega}(\sqrt{n}+D)$ for any approximation algorithm for directed diameter.
\end{proof}
